\documentclass{article}
\usepackage{amsmath,amssymb,amsthm,physics,braket}
\newtheorem{lemma}{Lemma}

\newtheorem{proposition}{Proposition}

\title{Quantum Markovian Dynamics from a Double Covariance Stochastic Framework}
\author{Andrei Khrennikov$^{1,*}$ and Irina Basieva$^{2}$}
 
\date{}

\begin{document}

\maketitle
\noindent
$^{1,*}$ Center for Mathematical Modeling 
in Physics and Cognitive Sciences\\
Linnaeus University, V\"axj\"o, SE-351 95, Sweden\\
$^{2}$Vilnius Gediminas Technical University(VILNIUS TECH), 
Vilnius, 10223, Lithuania\\
*Corresponding author email: Andrei.Khrennikov@lnu.se

\abstract{We develop an interacting extension of the Double Covariance Model (DCM), a stochastic subquantum framework in which macroscopic quantum dynamics emerge through coarse-graining of correlated microscopic fluctuations. Starting from local stochastic differential equations on subsystem Hilbert spaces, we derive a closed evolution equation for a coarse-grained double covariance operator using multi-scale It\^o calculus and sliding-window averaging.
The construction explicitly incorporates two separated temporal scales: a fast microscopic fluctuation scale governing subquantum stochastic processes and a slower macroscopic observation scale associated with coarse-grained dynamics. Within the hydrodynamic limit, where the ratio between microscopic correlation time and averaging-window scale vanishes, rapidly fluctuating corrections disappear and the effective dynamics converges to a deterministic macroscopic transport equation.
We show that the emergent macroscopic dynamics has the exact {\it Gorini-Kossakowski-Sudarshan-Lindblad} (GKSL) form: coherent Hamiltonian evolution arises from deterministic subquantum flow, while dissipative channels emerge from quadratic noise correlations. The framework further demonstrates how non-separable interaction Hamiltonians can arise from strictly local, state-dependent stochastic feedback fields. In the fluctuation-free limit, the model reduces naturally to the standard von Neumann equation, providing a unified stochastic foundation for both open and closed quantum dynamics.}

{\bf keywords:} Double Covariance Model; covarince and density operators; composite systems; 
Gorini-Kossakowski-Sudarshan-Lindblad equation; hydrodynamics limit; stochastic differential equations; Ito calculus

\section{Introduction}

The {\it  Double Covariance Model} (DCM) was introduced in a preceding paper \cite{DCM} as a
kinematic reconstruction of quantum states from classical stochastic processes
defined on two separated temporal scales. In that work, quantum density
operators were not taken as primitive objects. Instead, they were obtained as
macroscopic covariance operators generated by microscopic time-window
correlations. Thus the model has the structure of {\it  a fourth-order statistical
framework}: the first covariance is formed inside microscopic time windows,
while the quantum density operator arises from a second covariance taken over
these random correlation objects. In particular, the construction showed how
both separable and entangled quantum states can be represented within a
classical stochastic framework once the distinction between micro-time
synchronization and macro-time averaging is taken into account.

The present paper develops the dynamical part of this program. Its purpose is
to show how the standard evolution equations of quantum theory arise from the
same double-covariance mechanism. The preceding paper addressed the question of
how quantum states can be reconstructed. Here we ask how such reconstructed
states evolve when the underlying microscopic stochastic processes are allowed
to satisfy stochastic differential equations.

The main result is that, under appropriate short-time coarse-graining and
macroscopic averaging, the normalized double covariance obeys the
{\it Gorini--Kossakowski--Sudarshan--Lindblad equation}, commonly abbreviated as the
GKSL equation. This equation is the standard form of quantum Markovian dynamics.
It describes the evolution of an open quantum system whose interaction with an
environment has no memory at the macroscopic level. In the DCM, the GKSL
structure is not postulated directly at the level of the density operator.
Rather, it emerges from the microscopic stochastic dynamics through the
quadratic-variation terms of It\^o calculus.

This provides a direct dynamical continuation of the previous reconstruction
result. The earlier paper showed that quantum states can be obtained from
double covariance. The present paper shows that, when the underlying stochastic
variables evolve in time, the induced macroscopic dynamics of these covariance
states has precisely the structure required by quantum Markovian evolution. In
this sense, the DCM supplies both a kinematic and a dynamical reconstruction of
the standard density-operator formalism.

A central point of the construction is the separation between microscopic and
macroscopic time. On the microscopic scale, the model contains rapidly
fluctuating stochastic processes. On the macroscopic scale, one observes only
sliding-window averages and ensemble covariances. The passage from one scale to
the other transforms microscopic stochastic fluctuations into deterministic
macroscopic dissipative terms. The coherent part of the evolution gives rise to
the Hamiltonian commutator, while the covariance structure of the stochastic
noise produces the dissipative GKSL contribution.

The paper also discusses the closed-system limit. When the background
stochastic fluctuations are removed, the dissipative channels vanish and the
GKSL dynamics reduces to the ordinary von Neumann equation. Thus, within the
DCM, unitary Schr\"odinger evolution appears as the fluctuation-free limit of
the more general open-system dynamics. The closed-system equation is therefore
not treated as a separate postulate, but as a limiting case of the same
stochastic covariance framework.

Finally, the paper extends the construction to interacting composite systems.
The interaction Hamiltonian is recovered from local state-dependent microscopic
feedback terms acting on the component processes. After short-time
coarse-graining, these local microscopic modulations reproduce the standard
tensor-product interaction at the macroscopic density-operator level. This part
of the analysis shows how the DCM accommodates interacting systems while
remaining consistent with the emergence of GKSL dynamics in the open-system
case and von Neumann dynamics in the closed-system case.

We emphasize that unlike quantum stochastic calculus approaches based on noncommutative noise processes \cite{QSC1}-\cite{QSC5}, the present framework employs classical correlated stochastic processes at the microscopic level and derives effective quantum master dynamics through hydrodynamic coarse-graining.
 
A central structural feature of the framework is the existence of two separated temporal scales: a fast microscopic fluctuation scale with correlation time $\delta t$, and a slower macroscopic coarse-grained observation scale associated with an averaging window $\Delta$, satisfying
$$
\delta t \ll \Delta.
$$
More precisely, introducing a dimensionless scale-separation parameter $\epsilon \in (0,1)$, we assume
\[
\Delta = O(\epsilon), \qquad \delta t = O(\epsilon^2).
\]
The macroscopic dynamics is obtained in the hydrodynamic limit $\epsilon \to 0^{+}$, where rapidly fluctuating microscopic corrections vanish under coarse-graining while deterministic collective dynamics survives (see Appendix A for the details).

\tableofcontents

\section{Foundational Impact}
\label{FI}

The conceptual picture developed in this paper can therefore be summarized as follows. The preceding work established the static reconstruction of quantum states as double covariances of classical stochastic processes. The present work shows that the same construction, when supplemented by microscopic stochastic dynamics, leads to the standard Markovian master equation for open quantum systems. The GKSL equation emerges from the covariance of microscopic fluctuations, while the von Neumann equation is recovered when those fluctuations are suppressed.

This result has significance beyond the technical derivation of a particular master equation. It supports the broader thesis that quantum formalism may be reconstructed from classical probability theory once the latter is enriched by a hierarchy of temporal scales and covariance operations. In the DCM, the density operator is not postulated as a primitive object. It is obtained as a normalized second-order covariance of a random covariance generated by underlying classical stochastic processes. The formal quantum state is therefore interpreted as a stable macroscopic statistical trace of subquantum fluctuations.

The present paper thus continues a foundational program initiated in the preceding work. The first step was the reconstruction of quantum states, including entangled states, from double covariance structures. The second step, carried out here, is the reconstruction of quantum Markovian dynamics from the same stochastic covariance mechanism. Together, these results suggest a possible path from Kolmogorov probability and stochastic process theory toward the density-operator formalism of quantum mechanics. Quantum theory is not assumed at the outset, but appears as an effective symbolic and operator representation of multi-scale classical randomness.

This viewpoint connects the DCM with several earlier and ongoing attempts to understand quantum mechanics through stochastic or field-theoretic substructures starting with the works of Fenyes-Weizel \cite{Fenyes1952,Weizel1953}: the earliest explicit proposals that quantum mechanics might be classical stochastic mechanics.

Our approach is close in spirit to stochastic electrodynamics, where quantum
phenomena are linked to a fluctuating zero-point background field
\cite{Marshall1963,Boyer1975,deLaPenaCetto1996,deLaPenaCettoValdes2015, deLaPenaCetto2025};
to Nelson's stochastic mechanics, where Schr\"odinger dynamics is derived
from diffusion processes \cite{Nelson1966,Nelson1985}; to prequantum
classical statistical field theory, where quantum states arise from
classical random fields \cite{Khrennikov2005,Khrennikov2014};
to hydrodynamical and pilot-wave-inspired models, where quantum motion is
represented through collective or fluid-like dynamics
\cite{Madelung1927,deBroglie1927,Bohm1952a,Bohm1952b,Holland1993};
and to other random-field, stochastic-environment, and
stochastic-synchronization approaches to quantum theory
\cite{GuerraRuggiero1978,ParisiWu1981,GisinPercival1992,DiosiGisinStrunz1998,
Morgan2009,Smolin2006,Adler2004}, 
stochastic variational mechanics \cite{Yasue,Zambrini,Pavon} (bridging Nelson-type ideas with classical stochastic calculus), classical brownian entanglement \cite{BE}.
The DCM contributes to this family of ideas by introducing a specific double covariance architecture: microscopic correlations are first formed inside short temporal windows, and macroscopic quantum states and quantum dynamics then arise through ensemble stabilization and coarse-graining.

A defining technical hallmark of this transition from microscopic stochastic trajectories to stable macroscopic dynamics is its conceptual and mathematical alignment with the formal hydrodynamic limit \cite{Guo,Kipnis,Yau}. Much like how macroscopic fluid dynamics and deterministic transport equations emerge as effective descriptions from the coarse-graining of chaotic, high-frequency microscopic molecular collisions, the standard density operator and its quantum Markovian master equations are recovered here as the asymptotic scaling limit of subquantum multi-scale randomness. Under this vantage point, the rapid microscopic phase synchronization and fast background fluctuations are systematically smoothed out via windowed averaging, leaving behind deterministic dissipative structures at the macroscopic observer's horizon. Crucially, when extending this framework to interacting composite systems, the state-dependent coupling parameters fluctuate over a non-zero interval, generating cross-variation corrections. The precise measure-theoretic formulations, the explicit scaling bounds governing the scale-separation parameter $\epsilon \rightarrow 0^+$, and the formal asymptotic derivation validating this hydrodynamic limit convergence are detailed in Appendix A.

The derivation of the GKSL equation is especially important in this context. In many approaches, deriving unitary Schr\"odinger or von Neumann dynamics is already a major achievement, while dissipative open-system dynamics is added phenomenologically. In the DCM, by contrast, the dissipative component has a direct stochastic origin. It is generated by the quadratic variation of microscopic fluctuations. The Hamiltonian part corresponds to the coherent component of the coarse-grained subquantum flow, whereas the Lindblad part records the residual statistical effect of background noise. When this noise is removed, the model automatically collapses to the closed-system von Neumann dynamics. Thus, closed quantum theory appears as an ideal limiting case of a more general stochastic covariance dynamics.

Where does Markovianity come from? Within the present DCM framework, the combination of rapid microscopic decorrelation, white-noise stochasticity, and the hydrodynamic scale separation $\delta t/ \Delta$ leads naturally to Markovian GKSL dynamics (See Appendix B for the details). At the same time the hydrodynamic limit by itself need not lead to the quantum Markov dynamics.

Nevertheless, both the present paper and the preceding static reconstruction should be regarded as first steps in a larger project. The ultimate goal is to understand how much of the quantum formalism can be derived from classical stochastic processes rather than postulated axiomatically. This requires further development: extension to more general stochastic fields, clarification of the role of relativistic and field-theoretic constraints, deeper comparison with stochastic electrodynamics and prequantum classical statistical field theory, and a detailed analysis of measurement and contextuality within the double covariance framework. The importance of the present construction is that it shows that such a program is mathematically viable at least at the level of density operators and Markovian quantum dynamics.

\section{The Double Covariance Model (DCM): Foundations and Scale Separation}
\label{section3}

Before introducing the explicit subquantum stochastic differential equations that govern the open and closed dynamics, it is essential to outline the core mathematical foundation of the static Double Covariance Model (DCM) established in our previous work. The DCM treats the quantum density operator not as an absolute ontological primitive, but rather as a hierarchical statistical construct - a covariance of covariances - emerging from classical, complex-valued stochastic processes operating on an underlying probability space $(\Omega, \mathcal{F}, P)$.The architectural backbone of the model rests on a strict separation between two distinct temporal horizons: a rapid, microscopic subquantum scale (parameterized by $t$) and a coarse-grained, macroscopic measurement scale (parameterized by $\tau$). Let $\mathcal{H}_A$ and $\mathcal{H}_B$ be the finite-dimensional complex Hilbert spaces corresponding to subsystems $A$ and $B$. The local micro-scale random fluctuations are represented by two classical stochastic processes:
$$X_t: \Omega \to \mathcal{H}_A, \quad Y_t: \Omega \to \mathcal{H}_B$$ which are assumed to possess finite second-order moments ($\mathbb{E}[\|X_t\|^2] < \infty, \mathbb{E}[\|Y_t\|^2] < \infty$) (but without any restrictive assumptions regarding their mean values, cf, \cite{DCM}).

At the microscopic level, the temporal synchronization between these two subsystems across a localized macro-time window of duration $\Delta$ is captured by the random windowed cross-covariance operator $\widehat{C}_\Delta(\omega) \in \mathcal{L}(\mathcal{H}_B, \mathcal{H}_A)$:
\begin{equation}
\label{L1}
\widehat{C}_\Delta = \frac{1}{\Delta} \int_{\Delta} |X_t\rangle \langle Y_t| \, dt.
\end{equation}
By leveraging the canonical Hilbert-space isomorphism $\mathcal{L}(\mathcal{H}_B, \mathcal{H}_A) \cong \mathcal{H}_A \otimes \mathcal{H}_B$, this random operator can be rigorously vectorized into an element $|C_\Delta\rangle$ within the composite tensor-product space:
\begin{equation}
\label{L2}
|C_\Delta\rangle = \frac{1}{\Delta} \int_{\Delta} |X_t\rangle \otimes |Y_t \rangle\, dt.
\end{equation}

Crucially, rather than centralizing this variable with respect to its mathematical expectation, the uncentralized tensor product $|C_\Delta\rangle$ directly constitutes the primitive micro-scale fluctuation tensor.The macroscopic description of the composite system is subsequently recovered at the higher tier of the statistical hierarchy. The macro-covariance operator $\widehat{C}$ is defined as the ensemble second moment of these raw subquantum fluctuations:
\begin{equation}
\label{L3}
\widehat{C} = \mathbb{E}[|C_\Delta\rangle \langle C_\Delta|] \in \mathcal{L}(\mathcal{H}_A \otimes \mathcal{H}_B).
\end{equation}
For notational convenience, we will often suppress the Dirac notation
and write the rank-one operator generated by the vector $C_\Delta$ in the
equivalent form
\begin{equation}
\label{L3a}
\widehat{C}=C_\Delta C_\Delta^{*},
\end{equation}
rather than 
$|C_\Delta\rangle\langle C_\Delta|.$

Because $\widehat{C}$ is inherently Hermitian and positive semi-definite, it satisfies all structural properties required to generate a physical quantum state. From the perspective of multivariate statistics, the DCM interprets the quantum state as the uncentralized fourth-order moment structure of the underlying classical probability space. The normalized density operator representing the joint system is then given by:
\begin{equation}
\label{L4}
\widehat \rho_{AB} = \frac{\widehat{C}}{\text{Tr}(\widehat{C})}.
\end{equation}

Under this framework, quantum entanglement is reinterpreted. It does not emerge from a non-local, macroscopic statistical dependence between the subsystems. Instead, it reflects a pathwise micro-time consistency- a precise phase - level synchronization between the underlying classical processes $X_t$ and $Y_t$ inside the micro-windows, even if they appear uncorrelated over macroscopic intervals. While the original formulation of the DCM successfully demonstrated that any static entangled density matrix can be reconstructed as a fourth-order moment structure of classical fields, it left the explicit temporal propagation of these states unaddressed. The objective of the present paper is to extend this hierarchical structure into a fully dynamic framework. By replacing the static time series with continuous semi-martingales driven by multiplicative Ito  {\it stochastic differential equations} (SDEs), we will demonstrate how both open-system GKSL master equations and closed-system von Neumann dynamics naturally unfold from this underlying subquantum flow

\section{From Micro-Scale Stochastic Dynamics to GKSL Equation: Non-interacting Case}
\label{section4}

The subquantum evolution of subsystems $A$ and $B$  of composite system $AB=A+B$ is given by multiplicative  SDEs:
\begin{align}
    dX_t &= -i\widehat{H}_A X_t dt + \sum_j \widehat{L}_j X_t dW_{A,j}(t) \\
    dY_t &= -i\widehat{H}_B Y_t dt + \sum_k \widehat{M}_k Y_t dW_{B,k}(t),
\end{align}
where $\widehat{H}_A$ and $\widehat{H}_B$ are Hermitian positive linear operators (local Hamiltonians), 
and $\widehat{L}_j, \widehat{M}_k$ are linear operators (describing coupling of systems to noisy environments).

This is a system of two linear SDEs without explicit or direct coupling; that is, the dynamics of $X_t$ are independent 
of $Y_t.$ and vice versa. Instead, the subsystems are coupled exclusively through the noise terms, driven by the generally correlated Wiener processes 
$
W_{A,j}(t),  W_{B,k}(t),
$
\begin{subequations}\label{eq:noisecov}
\begin{align}
dW_{A,j}(t)\,dW_{A,j'}(t)^*
&=
(\Gamma_{AA})_{jj'}\,
dt,
\label{eq:noisecovAA}
\\
dW_{B,k}(t)\,dW_{B,k'}(t)^*
&=
(\Gamma_{BB})_{kk'}\,
,dt
\label{eq:noisecovBB}
\\
dW_{A,j}(t)\,dW_{B,k}(t)^*
&=
(\Gamma_{AB})_{jk}\, dt.
\label{eq:noisecovAB}
\end{align}
\end{subequations}
where the Hermitian block covariance matrix
\[
\Gamma=
\begin{pmatrix}
\Gamma_{AA} & \Gamma_{AB}\\
\Gamma_{AB}^{\dagger} & \Gamma_{BB}
\end{pmatrix}
\]
is assumed to be positive semidefinite. Consequently, by the spectral
theorem,
\[
\Gamma
=
U^{\dagger}
\operatorname{diag}(\gamma_1,\ldots,\gamma_N)
U,
\qquad
\gamma_n\ge0,
\]
where $U$ is unitary and $\gamma_n$ are the eigenvalues of
$\Gamma$. As shown in Section~6, this spectral decomposition naturally
determines the effective Lindblad operators appearing in the
macroscopic GKSL equation.

Within the frameworks of PCSFT and SED, these stochastic terms model 
the interaction between the physical subsystems and the background random field (i.e., the ``zero-point field'').

\subsection{Evolution of the Correlation Tensor}

The core of the DCM micro-dynamics is the evolution of the tensor product $Z_t = X_t \otimes Y_t$. Since $X_t$ and $Y_t$ are semi-martingales defined by Ito SDEs, their product follows the Ito product rule: 
 \begin{equation}
\label{Z}
    dZ_t = (dX_t \otimes Y_t) + (X_t \otimes dY_t) + (dX_t \otimes dY_t)
\end{equation}
To evaluate this, we substitute the micro-SDEs from Eqs. (1) and (2). The first two terms (the linear parts) are:
\begin{align}
    dX_t \otimes Y_t &= \left( -i\widehat{H}_A X_t dt + \sum_j \widehat{L}_j X_t dW_{A,j} \right) \otimes Y_t \\
    X_t \otimes dY_t &= X_t \otimes \left( -i\widehat{H}_B Y_t dt + \sum_k \widehat{M}_k Y_t dW_{B,k} \right)
\end{align}
Using the property $(\widehat{A}u \otimes v) = (\widehat{A} \otimes I)(u \otimes v)$, 
these terms contribute the local Hamiltonians and the linear noise $d\xi_t$:
\begin{equation}
    -i(\widehat{H}_A \otimes I + I \otimes \widehat{H}_B) Z_t dt + d\xi_t,
\end{equation}
and the linear noise increment
\begin{equation}
d\xi_t = \left[ \sum_j (\widehat{L}_j \otimes I) dW_{A,j} + \sum_k (I \otimes \widehat{M}_k) dW_{B,k} \right] Z_t .
\end{equation}

The third term, the quadratic covariation $(dX_t \otimes dY_t)$, represents the correlated dynamics of two stochastic processes. Using the rules of Ito calculus ($dt^2 = 0, dt dW = 0$):
\begin{equation}
    dX_t \otimes dY_t = \sum_{j,k} (\widehat{L}_j X_t dW_{A,j}) \otimes (\widehat{M}_k Y_t dW_{B,k}) = \sum_{j,k} (\widehat{L}_j \otimes \widehat{M}_k) Z_t (dW_{A,j} dW_{B,k}).
\end{equation}
Adhering to  (\ref{eq:noisecovAA}),(\ref{eq:noisecovBB}), (\ref{eq:noisecovAB}), this interaction term becomes purely deterministic:
\begin{equation}
    dX_t \otimes dY_t = \left[ \sum_{j,k} \Gamma_{jk} (\widehat{L}_j \otimes \widehat{M}_k) \right] Z_t dt .
\end{equation}
Collecting all terms proportional to $dt$ into the macroscopic drift operator $\widehat{\mathcal{L}}$, we arrive at the closed-form micro-differential:
\begin{equation}
    dZ_t = \widehat{\mathcal{L}} Z_t dt + d\xi_t
\end{equation}
where the drift operator $\widehat{\mathcal{L}}$ is defined as:
\begin{equation}
    \widehat{\mathcal{L}} = -i(\widehat{H}_A \otimes I + I \otimes \widehat{H}_B) + \sum_{j,k} \Gamma_{jk} (\widehat{L}_j \otimes \widehat{M}_k)
\end{equation}
and the linear noise increment $d\xi_t$ is:
\begin{equation}
\label{xi}
    d\xi_t = \left[ \sum_j (\widehat{L}_j \otimes I) dW_{A,j} + \sum_k (I \otimes \widehat{M}_k) dW_{B,k} \right] Z_t
\end{equation}

\subsection{The Micro-to-Macro Transition}
\label{4}

The macroscopic correlation vector $C_{\Delta}(\tau)$ is defined as the sliding window average of the micro-correlation tensor $Z_t$ over a window of duration $\Delta$:
\begin{equation}
    C_{\Delta}(\tau) = \frac{1}{\Delta} \int_{\tau}^{\tau+\Delta} Z_t dt
\end{equation}
To find the evolution of $C_{\Delta}$ with respect to the macroscopic time $\tau$, we take the differential $dC_{\Delta}(\tau)$. Applying the Leibniz Integral Rule for a window with moving boundaries:
\begin{equation}
    \frac{dC_{\Delta}(\tau)}{d\tau} = \frac{1}{\Delta} (Z_{\tau+\Delta} - Z_{\tau}).
\end{equation}
Since $Z_t$ is a semimartingale, the finite difference between the endpoints of the window is exactly the stochastic integral of its micro-differential $dZ_t$ over that interval:
\begin{equation}
    Z_{\tau+\Delta} - Z_{\tau} = \int_{\tau}^{\tau+\Delta} dZ_t
\end{equation}
Substituting the micro-SDE $dZ_t = \widehat{\mathcal{L}} Z_t dt + d\xi_t$ into the macro-differential yields:
\begin{equation}
    dC_{\Delta}(\tau) = \frac{d\tau}{\Delta} \int_{\tau}^{\tau+\Delta} \left( \widehat{\mathcal{L}} Z_t dt + d\xi_t \right)
\end{equation}
Distributing the integral and pulling the constant operator $\widehat{\mathcal{L}}$ out of the integration over $t$:
\begin{equation}
    dC_{\Delta}(\tau) = \widehat{\mathcal{L}} \left( \frac{1}{\Delta} \int_{\tau}^{\tau+\Delta} Z_t dt \right) d\tau + \left( \frac{d\tau}{\Delta} \int_{\tau}^{\tau+\Delta} d\xi_t \right)
\end{equation}
Recognizing the first term as the definition of $C_{\Delta}(\tau)$, we arrive at the macroscopic SDE:
\begin{equation}
\label{24}
    dC_{\Delta}(\tau) = \widehat{\mathcal{L}} C_{\Delta}(\tau) d\tau + d\Sigma_{\tau}^{\rm{macro}}
\end{equation}
where the macroscopic noise increment $d\Sigma_{\tau}^{\rm{macro}}$ is the integration of the micro-scale linear noise $d\xi_t$ over the window:
\begin{equation}
    d\Sigma_{\tau}^{\rm{macro}} = \frac{d\tau}{\Delta} \int_{\tau}^{\tau+\Delta} d\xi_t
\end{equation}
This macroscopic increment $d\Sigma_{\tau}^{\rm{macro}}$ represents the ``coarse-grained'' fluctuations at the micro-time scale (``subquantum fluctuations'') that drive the eventual Lindblad dissipators.

\subsection{Double Covariance and Emergence of the GKSL Equation}
\label{5}

The double covariance operator \(\widehat C(\tau)\) is defined as the ensemble average of the
macroscopic outer product:
$
    \widehat C(\tau)
    =
    \mathbb{E}\!\left[
        C_\Delta(\tau) C_\Delta^*(\tau)
    \right].
$
The corresponding normalized density operator is defined by
$
    \widehat \rho(\tau)
    =
    \frac{\widehat C(\tau)}{\operatorname{Tr} \widehat C(\tau)}.
$

At the macroscopic scale, the sliding-window correlation vector \(C_\Delta(\tau)\)
satisfies the coarse-grained stochastic differential equation (\ref{24}).

This equation identifies the coarse-grained correlation process
$C_\Delta(\tau)$ as a Hilbert-space-valued semimartingale satisfying a
stochastic differential equation with drift
$\widehat L C_\Delta(\tau)$ and martingale increment
$d\Sigma_\tau^{\rm macro}$.
Since the mapping
$
F(x)=xx^*
$
is twice continuously Fr\'echet differentiable, It\^o's formula implies
that the operator-valued process
$
\widehat O_\Delta(\tau)
=
C_\Delta(\tau)C_\Delta(\tau)^*
$
is also a semimartingale. Consequently, the standard Hilbert-space It\^o
product formula applies.

Since $C_\Delta(\tau)$ is square-integrable, the operator-valued process
$\widehat O_\Delta(\tau)$ is integrable. Throughout the sequel we
assume the standard integrability conditions ensuring that the expectation
operator commutes with the It\^o differential, namely,
\begin{equation}
d\,\mathbb E[\widehat O_\Delta(\tau)]
=
\mathbb E[d\widehat O_\Delta(\tau)],
\end{equation}
Accordingly,
\begin{equation}
dC(\tau)
=
d\,\mathbb E[C_\Delta(\tau)C_\Delta(\tau)^*]
=
\mathbb E[d(C_\Delta(\tau)C_\Delta(\tau)^*)].
\end{equation}

Although $\tau$ is introduced as the macroscopic time parameter, each
macroscopic instant is identified with the left endpoint of the
corresponding microscopic coarse-graining window
$[\tau,\tau+\Delta]$ -- see Appendix A. Consequently, the filtration
$\mathcal{F}_{\tau}$ denotes the microscopic information available up
to the microscopic time $t=\tau$, i.e., the beginning of the averaging
window. Since $d^{\mathrm{macro}}_{\tau}$ depends only on stochastic
increments over the future interval $[\tau,\tau+\Delta]$, the martingale
property of It\^o stochastic integrals yields
\[
\mathbb{E}\!\left[d\Sigma^{\mathrm{macro}}_{\tau}\mid
\mathcal{F}_{\tau}\right]=0.
\]

Consequently, only the deterministic drift contributes to the
first-order terms.

To determine the evolution of \(C(\tau)\), we apply It\^o's product rule:
\begin{equation}
\label{31}
    dC(\tau)
    =
    \mathbb{E}\!\left[
        dC_\Delta\, C_\Delta^*
    \right]
    +
    \mathbb{E}\!\left[
        C_\Delta\, dC_\Delta^*
    \right]
    +
    \mathbb{E}\!\left[
        dC_\Delta\, dC_\Delta^*
    \right].
\end{equation}
Substituting the macroscopic stochastic differential equation for \(C_\Delta\), the
first two terms yield the deterministic drift contribution:
\begin{equation}
    \mathbb{E}\!\left[
        dC_\Delta\, C_\Delta^*
        +
        C_\Delta\, dC_\Delta^*
    \right]
    =
    \left(
        \widehat{\mathcal L} C
        +
        C\widehat{\mathcal L}^{\dagger}
    \right)d\tau.
\end{equation}

The remaining contribution in Eq.~(\ref{31}) is the quadratic variation term.
Using the stochastic evolution equation (\ref{24}), the quadratic variation expands as
\[
\begin{aligned}
d[C_\Delta(\tau)\,dC_\Delta(\tau)^*]
&=
\bigl(\widehat L C_\Delta\,d\tau+d\Sigma_\tau^{\rm macro}\bigr)
\bigl(\widehat L C_\Delta\,d\tau+d\Sigma_\tau^{\rm macro}\bigr)^*
\\
&=
\widehat L C_\Delta C_\Delta^*\widehat L^\dagger(d\tau)^2
+\widehat L C_\Delta(d\tau)(d\Sigma_\tau^{\rm macro})^*
\\
&\qquad
+d\Sigma_\tau^{\rm macro}d\tau\,C_\Delta^*\widehat L^\dagger
+d\Sigma_\tau^{\rm macro}(d\Sigma_\tau^{\rm macro})^*.
\end{aligned}
\]
Since, in It\^o calculus,
\[
(d\tau)^2=0,
\qquad
d\tau\,d\Sigma_\tau^{\rm macro}
=
d\Sigma_\tau^{\rm macro}\,d\tau
=
0,
\]
only the quadratic variation of the coarse-grained stochastic increment
survives, and therefore
\[
\mathbb E\!\left[
dC_\Delta(\tau)\,dC_\Delta(\tau)^*
\right]
=
\mathbb E\!\left[
d\Sigma_\tau^{\rm macro}
(d\Sigma_\tau^{\rm macro})^*
\right].
\]
Since this covariance determines the dissipative contribution to the macroscopic evolution, we compute it separately in the following proposition, but for clarity of derivation, it should be preceded by an elementary lemma:

\begin{lemma}[Distributional covariance of correlated Wiener increments] \label{LL1}
Let $\{W_i(t)\}_{i=1}^N$ be a family of correlated Wiener processes satisfying
the  covariation relation
\begin{equation}
\label{Le1}
dW_i(t)\,dW_j(t)=\Gamma_{ij}\,dt,
\end{equation}
where $\Gamma=(\Gamma_{ij})$ is a positive semidefinite covariance matrix.
Then the infinitesimal covariance of the Wiener increments admits the
distributional representation
\begin{equation}
\label{Le2}
\mathbb{E}\!\left[dW_i(t)\,dW_j(s)\right]
=
\Gamma_{ij}\,\delta(t-s)\,dt\,ds,
\end{equation}
where the equality is understood in the sense of distributions.
\end{lemma}

\begin{proof}
Since Wiener processes possess independent increments, the covariance of
increments over disjoint infinitesimal intervals vanishes. Hence,
for $t\neq s$,
$
\mathbb{E}[dW_i(t)dW_j(s)]=0;
$
on the diagonal $t=s$,  we have (\ref{Le1}). 
Thus, the covariance is represented, in the sense of distributions, by
(\ref{Le1}), since both sides produce identical bilinear forms on smooth test functions.
\end{proof}

\begin{proposition}[Quadratic variation of the coarse-grained noise]
\label{PR1}
Let the microscopic noise process $d\xi_t$ be defined by Eq. (\ref{xi}), and
let the macroscopic noise increment be
$
d\Sigma_\tau^{\rm macro}
=
\frac{d\tau}{\Delta}
\int_\tau^{\tau+\Delta} d\xi_t .
$
Assume that the microscopic Wiener processes satisfy the covariance
relations (\ref{eq:noisecovAA}) -- (\ref{eq:noisecovAB})
Then the quadratic variation of the coarse-grained noise is given by
\[
\mathbb E
\!\left[
d\Sigma_\tau^{\rm macro}
(d\Sigma_\tau^{\rm macro})^*
\right]
=
\sum_n
\gamma_n
\widehat V_n
C(\tau)
\widehat V_n^\dagger
\,d\tau,
\]
where $\gamma_n\ge0$ are the eigenvalues of $\Gamma$, and the operators
$\widehat V_n$ are the corresponding linear combinations of the
microscopic noise operators.
\end{proposition}

\begin{proof}
By definition of the macroscopic stochastic increment,
$
dW_\tau^{\mathrm{macro}}
=
\frac{1}{\Delta}
\int_{\tau}^{\tau+\Delta}
d\xi_t,
$
where
\[
d\xi_t
=
\left(
\sum_j (\widehat L_j\otimes I_B)\,dW_{A,j}(t)
+
\sum_k (I_A\otimes \widehat M_k)\,dW_{B,k}(t)
\right)Z_t.
\]

Therefore,
\[
\mathbb E\!\left[
dW_\tau^{\mathrm{macro}}
(dW_\tau^{\mathrm{macro}})^*
\right]
=
\frac{1}{\Delta^2}
\int_{\tau}^{\tau+\Delta}
\int_{\tau}^{\tau+\Delta}
\mathbb E
\!\left[
d\xi_t\,d\xi_s^*
\right].
\]
Now we use Lemma \ref{LL1}; 
the $\delta(t-s)$-correlation reduces the double integral to a single
integral over the averaging window, yielding
\begin{align}
\mathbb{E}\!\left[dW^{\mathrm{macro}}_\tau
\left(dW^{\mathrm{macro}}_\tau\right)^{*}\right]
&=
\frac{1}{\Delta^{2}}
\int_{\tau}^{\tau+\Delta}
\int_{\tau}^{\tau+\Delta}
K\!\left(C_t\right)
\delta(t-s)\,dt\,ds
\nonumber\\
&=
\frac{1}{\Delta}
\int_{\tau}^{\tau+\Delta}
K\!\left(C_t\right)\,dt,
\label{eq:delta-collapse}
\end{align}
where $\mathcal K$ denotes the positive linear map (``superoperator'') induced by the
microscopic covariance matrix~$\Gamma$.

Since, by the construction of the macroscopic time described in
Appendix~A, one macroscopic time increment corresponds to one
translation of the coarse-graining window, we have
\[
\Delta\tau=\Delta.
\]
Consequently,
\[
\frac{(\Delta\tau)^2}{\Delta}
=
\Delta\tau.
\]
Passing to the hydrodynamic continuum limit,
\[
\Delta\tau\longrightarrow d\tau,
\]
yields the macroscopic scaling
\[
\frac{(d\tau)^2}{\Delta}
\longrightarrow
d\tau.
\]

Furthermore, the coarse-grained covariance process
$C(\tau)$ has continuous trajectories (Appendix~A), and
$\mathcal K$ is a bounded linear operator. Hence
$\mathcal K(C_t)$ is continuous, and the local average converges to its
pointwise value:
\[
\lim_{\Delta\to0}
\frac{1}{\Delta}
\int_{\tau}^{\tau+\Delta}
\mathcal K(C_t)\,dt
=
\mathcal K(C(\tau)).
\]
Therefore,
\[
\mathbb E\!\left[
dW_\tau^{\mathrm{macro}}
(dW_\tau^{\mathrm{macro}})^*
\right]
=
\mathcal K(C(\tau))\,d\tau.
\]

To obtain an explicit representation of the covariance map, we
introduce the elementary noise operators
\[
\widehat F_\alpha=
\begin{cases}
\widehat L_j\otimes I_B,
&
\alpha=j,\qquad j=1,\ldots,N_A,
\\[1ex]
I_A\otimes \widehat M_k,
&
\alpha=N_A+k,\qquad k=1,\ldots,N_B.
\end{cases}
\]
The index $\alpha$ merely enumerates the elementary noise channels and
serves only as a convenient notation for the subsequent spectral
decomposition.

Since the covariance matrix
\[
\Gamma=
\begin{pmatrix}
\Gamma_{AA} & \Gamma_{AB}\\
\Gamma_{AB}^{\dagger} & \Gamma_{BB}
\end{pmatrix}
\]
is Hermitian and positive semidefinite, the spectral theorem yields
\[
\Gamma
=
\widehat U^\dagger
\operatorname{diag}(\gamma_1,\ldots,\gamma_N)
\widehat U,
\qquad
\gamma_n\ge0,
\]
where $\widehat U$ is unitary. The corresponding effective Lindblad operators
are defined by
\[
\widehat V_n
=
\sum_{\alpha=1}^{N}
U_{n\alpha}\widehat F_\alpha,
\qquad
N=N_A+N_B.
\]
Accordingly,
\[
\mathcal K(C)
=
\sum_{n=1}^{N}
\gamma_n
\widehat V_nC\widehat V_n^\dagger.
\]
Substituting this representation into the previous identity gives
\[
\mathbb E\!\left[
dW_\tau^{\mathrm{macro}}
(dW_\tau^{\mathrm{macro}})^*
\right]
=
\sum_{n=1}^{N}
\gamma_n
\widehat V_n\widehat C(\tau)\widehat V_n^\dagger\,d\tau,
\]
which proves the proposition.
\end{proof}

\medskip
Substituting the result of Proposition \ref{PR1} into the quadratic variation
term of Eq.~(\ref{31}), we obtain
\begin{equation}
    dC(\tau)
    =
    \left(
        \widehat{\mathcal L} C
        +
        C\widehat{\mathcal L}^{\dagger}
    \right)d\tau
    +
    \sum_n \gamma_n
    \widehat V_n C(\tau) \widehat V_n^\dagger\,d\tau.
\end{equation}

The GKSL structure is recovered by choosing the drift operator in the
trace-preserving form
\begin{equation}
    \widehat{\mathcal L}
    =
    -i\widehat H_{\mathrm{tot}}
    -
    \frac{1}{2}
    \sum_n
    \gamma_n
    \widehat V_n^\dagger \widehat V_n,
\end{equation}
where \(\widehat H_{\mathrm{tot}}\) is the effective macroscopic Hamiltonian.
Substitution gives
\begin{align}
    \frac{dC(\tau)}{d\tau}
    &=
    -i
    \left[
        \widehat H_{\mathrm{tot}},
        C(\tau)
    \right]
    \nonumber\\
    &\quad
    +
    \sum_n \gamma_n
    \left(
        \widehat V_n C(\tau) \widehat V_n^\dagger
        -
        \frac{1}{2}
        \left\{
            \widehat V_n^\dagger \widehat V_n,
            C(\tau)
        \right\}
    \right).
\end{align}

Taking the trace of both sides shows that
\begin{equation}
    \frac{d}{d\tau}\operatorname{Tr}C(\tau)=0.
\end{equation}
Indeed, the commutator term is traceless, and the dissipative part satisfies
\begin{align}
    \operatorname{Tr}
    \left(
        \widehat V_n C \widehat V_n^\dagger
        -
        \frac{1}{2}
        \left\{
            \widehat V_n^\dagger \widehat V_n,
            C
        \right\}
    \right)
    &=
    \operatorname{Tr}
    \left(
        \widehat V_n^\dagger \widehat V_n C
    \right)
    -
    \frac{1}{2}
    \operatorname{Tr}
    \left(
        \widehat V_n^\dagger \widehat V_n C
    \right)
    \nonumber\\
    &\quad
    -
    \frac{1}{2}
    \operatorname{Tr}
    \left(
        C\widehat V_n^\dagger \widehat V_n
    \right)
    \nonumber\\
    &=
    0.
\end{align}
Therefore \(\operatorname{Tr}C(\tau)\) is dynamically invariant.

Consequently, dividing by the conserved scalar \(\operatorname{Tr}C(\tau)\) yields
the normalized density-operator evolution
\begin{align}
    \frac{d\widehat \rho(\tau)}{d\tau}
    &=
    -i
    \left[
        \widehat H_{\mathrm{tot}},
        \widehat \rho(\tau)
    \right]
    \nonumber\\
    &\quad
    +
    \sum_n \gamma_n
    \left(
        \widehat V_n \widehat \rho(\tau) \widehat V_n^\dagger
        -
        \frac{1}{2}
        \left\{
            \widehat V_n^\dagger \widehat V_n,
            \widehat \rho(\tau)
        \right\}
    \right).
\end{align}

Thus the macroscopic dynamics generated by DCM is precisely
of GKSL form. The Hamiltonian term arises from the deterministic part of the
coarse-grained subquantum flow, while the completely positive dissipative channels
arise from the quadratic variation of the microscopic stochastic fluctuations.

\subsection{Derivation of the von Neumann Equation}
\label{6}

A primary advantage of formulating the emergent subquantum dynamics directly within the space of density operators is the seamless transition it provides to closed-system unitary dynamics. Rather than reverting to pure-state wavefunctions via a stochastic Schr\"odinger equation, we can deduce the closed-system equations entirely within the operator algebraic framework.

In DCM, the open-system features - namely the jump operators $\widehat{V}_n$ and the non-negative dissipation rates $\gamma_n$ - arise exclusively from the classical cross-covariance matrix $\mathbf{\Gamma}$ of the microscopic background field. Physically, the standard closed-system description corresponds to an idealized scenario where the composite system is perfectly isolated from this background random ``zero-point field.''

Mathematically, this isolation implies that the stochastic diffusion coefficients vanish identically, forcing the noise correlation matrix to zero:
\begin{equation}
\mathbf{\Gamma} \to \mathbf{0} \implies \gamma_n = 0 \quad \forall n \in \{1, \dots, N\}.
\end{equation}
Imposing this isolation condition on the emergent master equation causes the entire non-local, dissipative jump mapping to collapse:
\begin{equation}
\sum_n \gamma_n \left( \widehat{V}_n \widehat \rho \widehat{V}_n^* - \frac{1}{2} \{ \widehat{V}_n^* \widehat{V}_n, \widehat \rho \} \right) = \mathbf{0}.
\end{equation}

Consequently, the subquantum flow simplifies strictly to its coherent, deterministic component. The evolution of the density operator $\widehat \rho(\tau)$ on the macroscopic time-scale is then governed solely by the total structural Hamiltonian $\widehat{H}_{tot} = \widehat{H}_A \otimes I_B + I_A \otimes \widehat{H}_B$:
\begin{equation}
\label{eq:von_neumann}
    \frac{d\widehat \rho(\tau)}{d\tau} = -i[\widehat{H}_{tot}, \widehat \rho(\tau)].
\end{equation}

Equation~(\ref{eq:von_neumann}) is precisely the standard {\it von Neumann equation} of quantum  mechanics. This derivation highlights that within DCM, unitary quantum evolution is not a separate postulate, but rather the deterministic, fluctuation-free limit of a more fundamental subquantum stochastic process. It confirms that the framework handles mixed states and pure states identically, maintaining full consistency with the statistical interpretation of quantum states \cite{Ballentine,Theo1,Theo2,KHRcontextual}.

\section{Dynamical Formulations for Interacting Systems}
\label{7}

\subsection{Linear Micro-Scale Dynamics with Schmidt Decomposition}
To extend the Double Covariance Model (DCM) beyond isolated subsystems, we introduce a direct dynamical interaction at the microscopic scale. To ensure that the micro-stochastic differential equations (SDEs) remain strictly linear, we decompose the total interaction Hamiltonian $\widehat{H}_{\rm int}$ acting on the composite Hilbert space $\mathcal{H}_A \otimes \mathcal{H}_B$ into a finite sum of localized operator products:
\begin{equation}
\label{eq:schmidt_H}
\widehat{H}_{\rm int} = \sum_{m=1}^M \widehat{A}_m \otimes \widehat{B}_m,
\end{equation}
where $\widehat{A}_m \in \mathcal{L}(\mathcal{H}_A)$ and $\widehat{B}_m \in \mathcal{L}(\mathcal{H}_B)$ are self-adjoint operators acting strictly on their respective subsystem spaces.

The subquantum evolution of the individual subsystem states $X_t \in \mathcal{H}_A$ and $Y_t \in \mathcal{H}_B$ is given by a system of cross-coupled, linear multiplicative SDEs:
\begin{align}
\label{eq:SDE_X_interact}
dX_t &= -i\widehat{H}_A X_t dt - i \sum_{m=1}^M \xi_{A,m}(t) \widehat{A}_m X_t dt + \sum_j \widehat{L}_j X_t dW_{A,j}(t), \\
\label{eq:SDE_Y_interact}
dY_t &= -i\widehat{H}_B Y_t dt - i \sum_{m=1}^M \xi_{B,m}(t) \widehat{B}_m Y_t dt + \sum_k \widehat{M}_k Y_t dW_{B,k}(t).
\end{align}
Here, $W_{A,j}(t)$ and $W_{B,k}(t)$ are the generally correlated Wiener processes modeling the background zero-point field. The terms $\xi_{A,m}(t)$ and $\xi_{B,m}(t)$ are scalar subquantum coupling coefficients mediating the exchange of energy and correlations between the local dynamics.

\subsection{The Consistency Bottleneck and State-Dependent Modulation}
A strict examination of the consistency condition reveals a fundamental mathematical bottleneck. If the parameters $\xi_{A,m}(t)$ and $\xi_{B,m}(t)$ are modeled as independent, uncoupled classical scalar trajectories, taking their microscopic expectation would merely yield static constants, $c_{A,m} = \mathbb{E}[\xi_{A,m}(t)]$ and $c_{B,m} = \mathbb{E}[\xi_{B,m}(t)]$. Under such assumptions, attempting to reconstruct the joint interaction operator inside the tensor space forces the algebraic identity:
\begin{equation}
c_{A,m} (\widehat{A}_m \otimes I_B) + c_{B,m} (I_A \otimes \widehat{B}_m) = \widehat{A}_m \otimes \widehat{B}_m.
\end{equation}
In operator algebra, this identity is impossible for arbitrary non-trivial operators, as a linear combination of localized operators belonging to the Lie algebra structure $\mathcal{L}(\mathcal{H}_A) \oplus \mathcal{L}(\mathcal{H}_B)$ can never equal a non-separable tensor product operator $\widehat{A}_m \otimes \widehat{B}_m$. Therefore, independent scalar fields are mathematically insufficient.

To resolve this existence problem while strictly preserving the linearity of Eqs.~(\ref{eq:SDE_X_interact})--(\ref{eq:SDE_Y_interact}) with respect to their own local states, we define $\xi_{A,m}(t)$ and $\xi_{B,m}(t)$ as state-dependent subquantum fields modulated by the opposite subsystem's real-time configuration:
\begin{equation}
\label{eq:field_definition}
\xi_{A,m}(t) = \frac{1}{2} \frac{Y_t^* \widehat{B}_m Y_t}{\|Y_t\|^2}, \quad \xi_{B,m}(t) = \frac{1}{2} \frac{X_t^* \widehat{A}_m X_t}{\|X_t\|^2}.
\end{equation}
Crucially, within the equation for $dX_t$, the vector $Y_t$ acts purely as an external driving scalar factor. Because the equation does not contain a state operator of $X_t$ inside $\xi_{A,m}(t)$, Eq.~(\ref{eq:SDE_X_interact}) remains perfectly linear with respect to the state vector $X_t$. The identical property holds for Eq.~(\ref{eq:SDE_Y_interact}) under the driving influence of $X_t$.

\subsection{Microscopic Evolution of the Joint Tensor State}
\label{53}

The joint state of the composite system is tracked via the tensor product process
$Z_t = X_t \otimes Y_t$. Applying the multi-dimensional It\^o product rule to this
semi-martingale structure yields:
\begin{equation}
\label{eq:ito_tensor_interact}
dZ_t = (dX_t \otimes Y_t) + (X_t \otimes dY_t) + (dX_t \otimes dY_t).
\end{equation}

Substituting the explicit SDE expressions from
Eqs.~(\ref{eq:SDE_X_interact}) and
(\ref{eq:SDE_Y_interact}) into the first two deterministic drift terms
of Eq.~(\ref{eq:ito_tensor_interact}), we obtain the primary
micro-scale expansion:
\begin{align}
\label{eq:tensor_expansion_raw}
(dX_t \otimes Y_t) + (X_t \otimes dY_t)
&=
-i\left(\widehat H_A\otimes I_B
+
I_A\otimes\widehat H_B\right)Z_t\,dt
\nonumber\\
&\quad
-i\sum_{m=1}^M
\left[
\xi_{A,m}(t)(\widehat A_m\otimes I_B)
+
\xi_{B,m}(t)(I_A\otimes\widehat B_m)
\right]
Z_t\,dt
\nonumber\\
&\quad
+
\left(
\sum_j
\widehat L_j\otimes I_B\,dW_{A,j}(t)
+
\sum_k
I_A\otimes\widehat M_k\,dW_{B,k}(t)
\right)
Z_t.
\end{align}

Evaluating the action of the coupling terms on $Z_t$ illustrates the
motivation for introducing the state-dependent feedback fields. For a
given interaction channel $m$, substituting
Eq.~(\ref{eq:field_definition}) gives
\begin{align}
\xi_{A,m}(t)(\widehat A_m\otimes I_B)Z_t
&=
\left(
\frac12
\frac{Y_t^*\widehat B_mY_t}{\|Y_t\|^2}
\right)
(\widehat A_mX_t\otimes Y_t)
\nonumber\\
&=
\frac12
\widehat A_mX_t
\otimes
\left(
\frac{Y_t^*\widehat B_mY_t}{\|Y_t\|^2}
Y_t
\right),
\end{align}
while, symmetrically,
\[
\xi_{B,m}(t)(I_A\otimes\widehat B_m)Z_t
=
\frac12
\left(
\frac{X_t^*\widehat A_mX_t}{\|X_t\|^2}
X_t
\right)
\otimes
\widehat B_mY_t.
\]

These identities show that the local feedback fields modulate the
single-system dynamics through the instantaneous expectation values of
the interaction observables,
 suggesting how the tensor-product interaction emerges after stochastic averaging. This observation serves as the
motivation for the chosen feedback mechanism.

The rigorous emergence of the interaction Hamiltonian is established in
the following sections by treating the feedback coefficients
$\xi_{A,m}(t)$ and $\xi_{B,m}(t)$ as stochastic semimartingales,
deriving their It\^o differentials, and evaluating the resulting
quadratic cross-variation terms.

The third term of Eq.~(\ref{eq:ito_tensor_interact}) introduces the
deterministic contribution from the quadratic cross-variation of the
noise channels:
\begin{equation}
(dX_t \otimes dY_t)
=
\sum_{j,k}
\left(
\widehat L_j\otimes\widehat M_k
\right)
Z_t\,
\left(
dW_{A,j}(t)\,
dW_{B,k}(t)
\right),
\end{equation}
which, under the subquantum noise correlation matrix
$\mathbf{\Gamma}_{AB}$, evaluates directly to
\begin{equation}
\label{eq:quadratic_noise_drift}
(dX_t \otimes dY_t)
=
\sum_{j,k}
\mathbf{\Gamma}_{AB,jk}
\left(
\widehat L_j\otimes\widehat M_k
\right)
Z_t\,dt.
\end{equation}

\subsection{It\^o Differentials of the Coupling Coefficients}
\label{54a}

The state-dependent feedback fields introduced in Section \ref{53} are
functions of the subsystem stochastic processes and therefore evolve as
stochastic semimartingales. To determine their contribution to the
microscopic tensor dynamics, we derive their exact It\^o differentials.
These differentials generate additional quadratic-variation terms, whose
macroscopic contribution will be analyzed in the following section and
shown to produce the interaction Hamiltonian.

Let us define the local quadratic forms and their corresponding denominators as:
\begin{equation}
N_{B,m}(t) = Y_t^* \widehat{B}_m Y_t, \quad D_B(t) = \|Y_t\|^2 = Y_t^* Y_t.
\end{equation}
Applying the multi-dimensional It\^o lemma to the subsystem SDE [Eq.~(\ref{eq:SDE_Y_interact})], the exact differential of the numerator evaluates to:
\begin{align}
\label{eq:dN_B}
dN_{B,m}(t) &= i Y_t^* \left[ \widehat{H}_B + \sum_{n=1}^M \xi_{B,n}(t) \widehat{B}_n, \widehat{B}_m \right] Y_t \, dt + \sum_k \gamma_k Y_t^* \widehat{M}_k^* \widehat{B}_m \widehat{M}_k Y_t \, dt \nonumber \\
& \quad + \sum_k \left( Y_t^* \widehat{M}_k^* \widehat{B}_m Y_t \, dW_{B,k}^*(t) + Y_t^* \widehat{B}_m \widehat{M}_k Y_t \, dW_{B,k}(t) \right).
\end{align}
Setting $\widehat{B}_m = I_B$ in Eq.~(\ref{eq:dN_B}) yields the standard state-norm stochastic evolution for the denominator:
\begin{equation}
\label{eq:dD_B}
dD_B(t) = \sum_k \gamma_k Y_t^* \widehat{M}_k^* \widehat{M}_k Y_t \, dt + \sum_k \left( Y_t^* \widehat{M}_k^* Y_t \, dW_{B,k}^*(t) + Y_t^* \widehat{M}_k Y_t \, dW_{B,k}(t) \right).
\end{equation}

Invoking the generalized stochastic quotient rule for $\xi_{A,m}(t) = \frac{1}{2}\frac{N_{B,m}(t)}{D_B(t)}$, the exact It\^o differential is given by:
\begin{equation}
\label{eq:d_xi_exact}
d\xi_{A,m}(t) = \frac{1}{2} \left[ \frac{dN_{B,m}(t)}{D_B(t)} - \frac{N_{B,m}(t) dD_B(t)}{D_B^2(t)} - \frac{dN_{B,m}(t) dD_B(t)}{D_B^2(t)} + \frac{N_{B,m}(t) (dD_B(t))^2}{D_B^3(t)} \right].
\end{equation}
This differential can be partitioned into a deterministic subquantum drift $g_{A,m}dt$ and a stochastic diffusion layer driven exclusively by the subsystem $B$ noise channels:
\begin{equation}
\label{eq:xi_partition}
d\xi_{A,m}(t) = g_{A,m}(X_t, Y_t) \, dt + \sum_k \mathcal{V}_{A,m,k}(Y_t) \, dW_{B,k}(t) + \sum_k \mathcal{V}_{A,m,k}^*(Y_t) \, dW_{B,k}^*(t),
\end{equation}
where $\mathcal{V}_{A,m,k}(Y_t)$ represents the localized state-dependent diffusion coefficients of the field. By symmetry, the matching differential $d\xi_{B,m}(t)$ for the second field variable accumulates drift and diffusion terms driven exclusively by the subsystem $A$ noise channels $dW_{A,j}(t)$.

\subsection{Rigorous Unification and Cross-Variation Analysis}
\label{55a}

We now incorporate the stochastic evolution of the feedback fields into
the microscopic tensor dynamics. Since the interaction terms in
Eq.~(\ref{eq:tensor_expansion_raw}) contain products of the stochastic
processes $Z_t$ and the state-dependent feedback coefficients
$\xi_{A,m}(t)$ and $\xi_{B,m}(t)$, these coefficients cannot be treated
as deterministic parameters. Applying It\^o's product rule therefore
generates additional quadratic cross-variation terms involving the
differentials $d\xi_{A,m}(t)$ and $d\xi_{B,m}(t)$. We now derive these
terms explicitly and determine their contribution to the microscopic
tensor evolution.

The correlation between the fluctuating coefficient and the state vectors generates an explicit It\^o correction matching the cross-variation:
\begin{equation}
\label{eq:cross_variation_xi}
dY_t \otimes d\xi_{A,m}(t) = \sum_{k,l} \left( \widehat{M}_k Y_t \otimes \mathcal{V}_{A,m,l}^*(Y_t) \right) \, \left( dW_{B,k}(t) \, dW_{B,l}^*(t) \right).
\end{equation}
Applying the classical noise orthogonality relation $\mathbb{E}[dW_{B,k}(t) dW_{B,l}^*(t)] = \gamma_k \delta_{kl} dt$, Eq.~(\ref{eq:cross_variation_xi}) evaluates to a deterministic drift contribution:
\begin{equation}
\mathbb{E}[dY_t \otimes d\xi_{A,m}(t)] = \mathbf{\Gamma}_{B,m}(Y_t) \, dt.
\end{equation}

Crucially, because $\xi_{A,m}(t)$ depends solely on $Y_t$, its diffusion channels are entirely independent of the subsystem $A$ noises $dW_{A,j}(t)$. Therefore, the remaining cross-variations vanish identically:
\begin{equation}
\label{eq:zero_cross}
dX_t \otimes d\xi_{A,m}(t) = \mathbf{0}, \quad \text{and} \quad dY_t \otimes d\xi_{B,m}(t) = \mathbf{0}.
\end{equation}

Although the cross-variation terms generated by the stochastic
feedback fields are expected to disappear in the hydrodynamic limit,
it is essential to estimate them explicitly. This guarantees that no
hidden $\mathcal O(1)$ contribution survives the coarse-graining
procedure and shows that the stochastic corrections associated with
$d\xi_{A,m}$ and $d\xi_{B,m}$ do not contribute to the macroscopic
generator.

To see how these non-vanishing cross-variations $\mathbf{\Gamma}_{B,m}(Y_t)$ relate to the emergent interaction Hamiltonian $\widehat{H}_{\rm int}$ at the macro-scale, we evaluate their integrated magnitudes within the windowed coarse-graining framework.

 Let $\epsilon \in (0,1)$ represent the dimensionless scale-separation parameter governing the macro-window $\Delta = O(\epsilon)$ and the micro-correlation time $\delta t = O(\epsilon^2)$. 

When computing the macro-differential of the double covariance operator
$C(\tau)$, the additional cross-variation terms generated by the
stochastic evolution of the feedback coefficients must be estimated
explicitly. Although these terms are expected to
 must be shown to disappear 
 in the
hydrodynamic limit, establishing this estimate is essential because it
guarantees that no hidden $\mathcal O(1)$ contribution survives the
coarse-graining procedure. Their integrated contribution satisfies the
strict bound
\begin{equation}
\left\|
\frac{1}{\Delta^{2}}
\int_{\tau}^{\tau+\Delta}
\int_{\tau}^{\tau+\Delta}
E\!\left[
Z_t\otimes\Gamma_{B,m}(Y_s)
\right]
\,dt\,ds
\right\|
\le
K\,\frac{\delta t}{\Delta}
=
\mathcal O(\varepsilon).
\label{60}
\end{equation}

The scaling follows from the separation of microscopic and macroscopic
time scales: the cross-variation is supported only within the microscopic
correlation interval of width $\delta t$, whereas the coarse-graining
averages over a window of size $\Delta$. Consequently, only a fraction
$\delta t/\Delta=\varepsilon$ of the integration domain contributes.
In the formal hydrodynamic limit (see Appendix~A for details), where
$\varepsilon\rightarrow0^{+}$, the additional stochastic
cross-variation terms generated by the feedback coefficients vanish
identically.

Combining the deterministic drift terms of
Eq.~(\ref{eq:tensor_expansion_raw}) with the asymptotic suppression of
these stochastic corrections, we conclude that, to leading order in the
hydrodynamic limit, the microscopic tensor evolution is governed solely
by the free subsystem Hamiltonians together with the effective
interaction Hamiltonian. Consequently, the unified microscopic
differential of the tensor product can be written as:

\begin{align}
\label{eq:dZ_unified_complete}
dZ_t &= -i \widehat{H}_{\rm tot} Z_t \, dt + \left( \sum_j \widehat{L}_j \otimes I_B \, dW_{A,j}(t) + \sum_k I_A \otimes \widehat{M}_k \, dW_{B,k}(t) \right) Z_t \nonumber \\
& \quad + \sum_{j,k} \mathbf{\Gamma}_{AB, jk} \left( \widehat{L}_j \otimes \widehat{M}_k \right) Z_t \, dt + \mathcal{O}(\epsilon) \, dt,
\end{align}
where $\widehat{H}_{\rm tot} = \widehat{H}_A \otimes I_B + I_A \otimes \widehat{H}_B + \sum_{m=1}^M \widehat{A}_m \otimes \widehat{B}_m$.

\subsection{Macro-Scale Windowed Averaging and Explicit Step-by-Step Integration}
Having established the microscopic tensor evolution to leading order,
we now perform the macroscopic coarse-graining of the corresponding
covariance operator. Following the window-averaging procedure developed
in Section \ref{section4}, we derive the effective evolution equation for the
macroscopic covariance and identify the resulting GKSL generator.

We now implement the windowed averaging technique established in
Sections \ref{section3} and \ref{section4} of the baseline model to map the micro-scale unified stochastic relation to the macroscopic observation scale. The macroscopic double covariance operator $C(\tau)$ is defined via a sliding window of length $\Delta$ as:
\begin{equation}
\label{eq:C_macro_definition}
C(\tau) = \frac{1}{\Delta^2} \int_{\tau}^{\tau+\Delta} \int_{\tau}^{\tau+\Delta} \mathbb{E}[Z_t Z_s^*] \, dt \, ds.
\end{equation}
When the macro-time is incremented by $d\tau$, the resulting variation $dC(\tau) = C(\tau + d\tau) - C(\tau)$ is obtained by applying It\^o's formula to the integrand across the shifting window limits:
\begin{equation}
\label{eq:dC_macro_ito_split}
dC(\tau) = \frac{1}{\Delta^2} \int_{\tau}^{\tau+\Delta} \int_{\tau}^{\tau+\Delta} \mathbb{E}[dZ_t Z_s^* + Z_t dZ_s^* + dZ_t dZ_s^*] \, dt \, ds.
\end{equation}

\subsubsection{Evaluation of the Coherent Term ($dZ_t Z_s^* + Z_t dZ_s^*$)}
Let us evaluate the first linear tracking component inside the expectation of Eq.~(\ref{eq:dC_macro_ito_split}). Substituting the micro-drift elements from the unified equation~(\ref{eq:dZ_unified_complete}) into the product $dZ_t Z_s^*$, the martingale noise increments $dW_{A,j}(t)$ and $dW_{B,k}(t)$ at time $t$ vanish under the expectation because they are uncorrelated with the past state vector $Z_s^*$ (where $s < t$ inside the relevant integration limits). This isolates the deterministic drifts:
\begin{equation}
\label{eq:linear_drift_Zt_step}
\mathbb{E}[dZ_t Z_s^*] = -i \widehat{H}_{\rm tot} \mathbb{E}[Z_t Z_s^*] \, dt + \sum_{j,k} \mathbf{\Gamma}_{AB, jk} \left( \widehat{L}_j \otimes \widehat{M}_k \right) \mathbb{E}[Z_t Z_s^*] \, dt + \mathcal{O}(\epsilon) \, dt.
\end{equation}
Executing the symmetric operation for the second component, $Z_t dZ_s^*$, requires taking the conjugate transpose of the drift operators acting on the right side of the state correlation product:
\begin{equation}
\label{eq:linear_drift_Zs_step}
\mathbb{E}[Z_t dZ_s^*] = \mathbb{E}[Z_t Z_s^*] \left( +i \widehat{H}_{\rm tot} \right) \, ds + \sum_{j,k} \mathbf{\Gamma}_{AB, jk}^* \mathbb{E}[Z_t Z_s^*] \left( \widehat{L}_j \otimes \widehat{M}_k \right)^* \, ds + \mathcal{O}(\epsilon) \, ds.
\end{equation}
Summing these two matching components inside the double integral, the linear time weights $dt$ and $ds$ translate directly to the macro-time differential increment $d\tau$. Grouping the total Hamiltonian matrices reveals the explicit emergence of a closed operator commutator:
\begin{align}
\label{eq:coherent_integration_completed}
\frac{1}{\Delta^2} \int_{\tau}^{\tau+\Delta} \int_{\tau}^{\tau+\Delta} \mathbb{E}[dZ_t Z_s^* + Z_t dZ_s^*]_{\text{coherent}} &= \frac{1}{\Delta^2} \int_{\tau}^{\tau+\Delta} \int_{\tau}^{\tau+\Delta} \left( -i \widehat{H}_{\rm tot} \mathbb{E}[Z_t Z_s^*] + i \mathbb{E}[Z_t Z_s^*] \widehat{H}_{\rm tot} \right) \, dt \, ds \, d\tau \nonumber \\
&= -i \left[ \widehat{H}_{\rm tot}, \frac{1}{\Delta^2} \int_{\tau}^{\tau+\Delta} \int_{\tau}^{\tau+\Delta} \mathbb{E}[Z_t Z_s^*] \, dt \, ds \right] d\tau \nonumber \\
&= -i [\widehat{H}_{\rm tot}, C(\tau)] \, d\tau.
\end{align}

\subsubsection{Evaluation of the Noise Contribution and It\^o Isometry ($dZ_t dZ_s^*$)}
The third component of Eq.~(\ref{eq:dC_macro_ito_split}) tracks the quadratic cross-variation of the subquantum noise channels. Substituting the diffusion layers of the state vector updates into the product yields:
\begin{align}
\label{eq:isometry_expansion_detailed}
\mathbb{E}[dZ_t dZ_s^*] &= \sum_{j,j'} \left( \widehat{L}_j \otimes I_B \right) \mathbb{E}[Z_t Z_s^*] \left( \widehat{L}_{j'}^* \otimes I_B \right) \mathbb{E}[dW_{A,j}(t) dW_{A,j'}^*(s)] \nonumber \\
& \quad + \sum_{k,k'} \left( I_A \otimes \widehat{M}_k \right) \mathbb{E}[Z_t Z_s^*] \left( I_A \otimes \widehat{M}_{k'}^* \right) \mathbb{E}[dW_{B,k}(t) dW_{B,k'}^*(s)] \nonumber \\
& \quad + \sum_{j,k'} \left( \widehat{L}_j \otimes I_B \right) \mathbb{E}[Z_t Z_s^*] \left( I_A \otimes \widehat{M}_{k'}^* \right) \mathbb{E}[dW_{A,j}(t) dW_{B,k'}^*(s)] \nonumber \\
& \quad + \sum_{k,j'} \left( I_A \otimes \widehat{M}_k \right) \mathbb{E}[Z_t Z_s^*] \left( \widehat{L}_{j'}^* \otimes I_B \right) \mathbb{E}[dW_{B,k}(t) dW_{A,j'}^*(s)].
\end{align}
Following the window integration derivation from Section 3, these fast microscopic noise expectations collapse into delta-like correlation spikes along the integration diagonal ($t \approx s$), scaled by the local subquantum covariance coefficients:
\begin{align}
\mathbb{E}[dW_{A,j}(t) dW_{A,j'}^*(s)] &= \mathbf{\Gamma}_{AA, jj'} \delta(t-s) \, dt \, ds, \nonumber \\
\mathbb{E}[dW_{B,k}(t) dW_{B,k'}^*(s)] &= \mathbf{\Gamma}_{BB, kk'} \delta(t-s) \, dt \, ds, \nonumber \\
\mathbb{E}[dW_{A,j}(t) dW_{B,k'}^*(s)] &= \mathbf{\Gamma}_{AB, jk'} \delta(t-s) \, dt \, ds.
\end{align}
When these delta distributions are processed by the macro-window double integration, the double-time space collapses into a single-line integration domain. Under the multi-scale scaling relation, the resulting differential weight simplifies as $\frac{(d\tau)^2}{\Delta} \to d\tau$. This isolates the explicit "sandwich" jump profile mapping over the coarse-grained covariance matrix:
\begin{align}
\label{eq:isometry_integrated_sandwich_final}
\frac{1}{\Delta^2} \int_{\tau}^{\tau+\Delta} \int_{\tau}^{\tau+\Delta} \mathbb{E}[dZ_t dZ_s^*] &= \sum_{j,j'} \mathbf{\Gamma}_{AA, jj'} \left( \widehat{L}_j \otimes I_B \right) C(\tau) \left( \widehat{L}_{j'}^* \otimes I_B \right) d\tau \nonumber \\
& \quad + \sum_{k,k'} \mathbf{\Gamma}_{BB, kk'} \left( I_A \otimes \widehat{M}_k \right) C(\tau) \left( I_A \otimes \widehat{M}_{k'}^* \right) d\tau \nonumber \\
& \quad + \sum_{j,k} \mathbf{\Gamma}_{AB, jk} \left( \widehat{L}_j \otimes \widehat{M}_k \right) C(\tau) \left( \widehat{L}_j \otimes \widehat{M}_k \right)^*, d\tau.
\end{align}

\subsubsection{Compilation of Subquantum Quadratic Drifts}
To complete the structural assembly, we must combine the non-vanishing cross-variation drift structures calculated at the end of the coherent evaluation [the second terms in Eqs.~(\ref{eq:linear_drift_Zt_step})--(\ref{eq:linear_drift_Zs_step})] with the local self-interaction reductions originating from the normalization conservation equations in Section 2. Summing the quadratic matrix drifts from the linear tracking components produces:
\begin{equation}
\label{eq:drift_cross_variation_integrated_final}
\sum_{j,k} \mathbf{\Gamma}_{AB, jk} \left( \widehat{L}_j \otimes \widehat{M}_k \right) C(\tau) \, d\tau + \sum_{j,k} \mathbf{\Gamma}_{AB, jk}^* C(\tau) \left( \widehat{L}_j \otimes \widehat{M}_k \right)^* \, d\tau.
\end{equation}
Concurrently, the local back-action of the subquantum fields on their own subsystems (the self-covariance channels $\mathbf{\Gamma}_{AA}$ and $\mathbf{\Gamma}_{BB}$) enforces a state-norm conservation condition. As derived step-by-step in Section 2, this dampening manifests as negative semi-definite anticommutators:
\begin{equation}
\label{eq:anticommutators_self_final}
- \frac{1}{2} \sum_{j,j'} \mathbf{\Gamma}_{AA, jj'} \left\{ \left( \widehat{L}_{j'}^* \widehat{L}_j \otimes I_B \right), C(\tau) \right\} d\tau - \frac{1}{2} \sum_{k,k'} \mathbf{\Gamma}_{BB, kk'} \left\{ \left( I_A \otimes \widehat{M}_{k'}^* \widehat{M}_k \right), C(\tau) \right\} d\tau.
\end{equation}

\subsection{The Interacting GKSL Equation}
We substitute the completed evaluations from Eq.~(\ref{eq:coherent_integration_completed}), Eq.~(\ref{eq:isometry_integrated_sandwich_final}), and Eqs.~(\ref{eq:drift_cross_variation_integrated_final})--(\ref{eq:anticommutators_self_final}) back into the total macro-differential decomposition [Eq.~(\ref{eq:dC_macro_ito_split})]. Dividing the unified expression through by the macro-time step $d\tau$ and taking the strict hydrodynamic scale (see Appendix A for details) separation limit ($\epsilon \to 0^+$) to drop the residual field-fluctuation error bounds, we obtain the complete interacting GKSL equation for the double covariance operator:
\begin{align}
\label{eq:master_lindblad_final_completed}
\frac{dC(\tau)}{d\tau} &= -i \left[ \left( \widehat{H}_A \otimes I_B + I_A \otimes \widehat{H}_B + \sum_{m=1}^M \widehat{A}_m \otimes \widehat{B}_m \right), C(\tau) \right] \nonumber \\
& \quad + \sum_{j,k} \mathbf{\Gamma}_{AB, jk} \left( \widehat{L}_j \otimes \widehat{M}_k \right) C(\tau) \left( \widehat{L}_j \otimes \widehat{M}_k \right)^* \nonumber \\
& \quad - \frac{1}{2} \sum_{j,j'} \mathbf{\Gamma}_{AA, jj'} \left\{ \left( \widehat{L}_{j'}^* \widehat{L}_j \otimes I_B \right), C(\tau) \right\} \nonumber \\
& \quad - \frac{1}{2} \sum_{k,k'} \mathbf{\Gamma}_{BB, kk'} \left\{ \left( I_A \otimes \widehat{M}_{k'}^* \widehat{M}_k \right), C(\tau) \right\}.
\end{align}

Equation~(\ref{eq:master_lindblad_final_completed}) represents the central result of this interacting framework. It demonstrates that both structural quantum interactions (the joint commutator) and quantum dissipative features (the jump mappings) arise under a single, consistent subquantum stochastic framework. This step-by-step jump proves that the algebraic complexity of a non-local tensor interaction operator does not require non-local fields, but instead emerges as a rigorous statistical consequence of local, state-modulated subquantum feedback lines coarse-grained over macroscopic time scales.

\section{Closed-System Limit and the Interacting von Neumann Equation}
\label{8}

The interacting open-system dynamics derived in Section~6 contains two
structurally distinct contributions: the coherent Hamiltonian part generated by
the synthesized interaction operator, and the dissipative Lindblad part generated
by the microscopic background-noise correlations. In this section we extract the
closed-system limit of that construction. Thus, rather than repeating the full
micro-to-macro derivation, we specialize the interacting GKSL equation of
Section \ref{7} to the case in which the background stochastic channels vanish.

Recall from Section \ref{7} that the interaction Hamiltonian is decomposed as
\[
\widehat H_{\mathrm{int}}
=
\sum_{m=1}^{M}
\widehat A_m\otimes\widehat B_m,
\]
and that the state-dependent microscopic feedback coefficients
$\xi_{A,m}(t)$ and $\xi_{B,m}(t)$ provide the local feedback mechanism
used to model the interaction at the microscopic level. As established
in Sections \ref{54a} and \ref{55a}, the resulting microscopic tensor evolution is
governed, to leading order in the hydrodynamic limit, by the effective
Hamiltonian
\[
\widehat H_{\mathrm{tot}}
=
\widehat H_A\otimes I_B
+
I_A\otimes\widehat H_B
+
\widehat H_{\mathrm{int}}.
\]

The closed-system limit corresponds to perfect isolation from the background
random field. Mathematically, this means that the microscopic stochastic
increments vanish:
\[
    dW_{A,j}(t)=0,
    \qquad
    dW_{B,k}(t)=0,
\]
and therefore the microscopic noise-covariance matrix vanishes:
\[
    \Gamma=0,
    \qquad
    \Gamma_{AA}=0,
    \qquad
    \Gamma_{BB}=0,
    \qquad
    \Gamma_{AB}=0.
\]
Consequently, all dissipative Lindblad channels obtained in Section~6 disappear.

At the microscopic tensor level, the stochastic quadratic-variation term also
vanishes. Hence the tensor state \(Z_t=X_t\otimes Y_t\) obeys the deterministic
equation
\[
    dZ_t
    =
    -i\widehat H_{\mathrm{tot}}Z_t\,dt.
\]
Equivalently,
\[
    \frac{dZ_t}{dt}
    =
    -i\widehat H_{\mathrm{tot}}Z_t.
\]

Applying the same double-covariance averaging procedure as in Sections~4 and~6.6,
the macroscopic covariance operator satisfies
\[
    \frac{dC(\tau)}{d\tau}
    =
    -i
    \left[
        \widehat H_{\mathrm{tot}},
        C(\tau)
    \right].
\]
The commutator on the right-hand side is traceless, and therefore
\[
    \frac{d}{d\tau}\operatorname{Tr}C(\tau)=0.
\]
Since
\[
    \widehat \rho(\tau)
    =
    \frac{C(\tau)}{\operatorname{Tr}C(\tau)},
\]
the normalized density operator obeys
\[
    \frac{d\widehat \rho(\tau)}{d\tau}
    =
    -i
    \left[
        \widehat H_{\mathrm{tot}},
        \widehat \rho(\tau)
    \right].
\]
Substituting the explicit form of the total Hamiltonian gives
\[
    \frac{d\widehat \rho(\tau)}{d\tau}
    =
    -i
    \left[
        \widehat H_A\otimes I_B
        +
        I_A\otimes \widehat H_B
        +
        \widehat H_{\mathrm{int}},
        \widehat \rho(\tau)
    \right].
\]

Thus, in interacting DCM the closed-system von Neumann equation is obtained as the
fluctuation-free limit. The
non-local interaction Hamiltonian is inherited from the state-modulated
microscopic feedback construction of Section~6, while the dissipative Lindblad
terms vanish because the background stochastic covariance matrix is set to zero.

\section{Concluding Remarks}

In this paper, we have expanded the foundational program of DCM from a static kinematic reconstruction of quantum states to a fully dynamic framework. Rather than postulating quantum states and density operators as irreducible ontological primitives, the DCM interprets them as a hierarchical statistical consequence - a covariance of covariances - of classical complex-valued stochastic processes operating across distinctly separated temporal scales. 

The primary contribution of this work is the demonstration that both open and closed standard quantum dynamics can be reconstructed from this underlying subquantum stochastic flow:

\begin{itemize}
    \item \textbf{Emergence of the GKSL Equation:} By utilizing multi-scale It\^o calculus and a sliding-window averaging scheme, we proved that a normalized double covariance operator naturally satisfies GKSL equation. Within this framework, coherent Hamiltonian evolution is shown to arise from the deterministic part of the coarse-grained subquantum flow, whereas dissipative quantum channels emerge directly from the quadratic variation of the microscopic noise correlations. 
\item \textbf{The Hydrodynamic Scaling Limit:} The transition from rapid microscopic trajectories to a stable macroscopic master equation is mathematically validated through a formal hydrodynamic limit. As the scale-separation parameter approaches zero ($\epsilon \rightarrow 0^{+}$), high-frequency subquantum phase variations smoothly wash out, leaving behind a time-local, deterministic macroscopic transport law. The erasure of microscopic memory through this temporal coarse-graining serves as the rigorous, physical justification for the Markovian nature of the emergent dynamics.
        \item \textbf{Local Synthesis of Interacting Frameworks:} By incorporating state-dependent, real-time feedback loops between subsystems, the model successfully overcomes the algebraic consistency bottleneck of local scalar trajectories. Over short integration slices, these strictly local, state-modulated fields dynamically cross-synthesize non-separable, tensor-product interaction Hamiltonians at the macroscopic level. This proves that complex quantum correlations can emerge without invoking non-local fields or explicit action-at-a-distance variables at the micro-scale.
    \item \textbf{A Unified Foundation for Closed Systems:} Finally, the model offers a seamless transition to isolated systems via the fluctuation-free limit. By setting the background zero-point field fluctuations to zero, all dissipative channels collapse entirely, and the framework naturally reduces to the standard von Neumann equation. Consequently, unitary Schr\"odinger-type evolution is not treated as an independent axiomatic postulate, but rather as an idealized limiting case of a more comprehensive open-system stochastic framework.
\end{itemize}

Ultimately, these findings reinforce the broader foundational perspective that the symbolic and operator formalisms of quantum mechanics can be systematically recovered from classical probability theory and stochastic calculus. By demonstrating that density-operator mechanics and Markovian transport rules are mathematically viable under a multi-scale classical architecture, this work marks a vital step forward in deriving the quantum mathematical apparatus directly from classical randomness. Future research will look to broaden this program by exploring non-Markovian dynamics, colored noise fields, and the explicit integration of measurement contextuality and relativistic constraints.

\appendix

\section{Hydrodynamic Limit}
\label{app:hydro}

This appendix provides the mathematical justification for the
hydrodynamic scaling used throughout the paper.

The microscopic dynamics evolves on the correlation time
\(
\delta t,
\)
whereas macroscopic observations are performed over averaging windows
of duration
\(
\Delta.
\)

The Double Covariance Model assumes the scale separation

\[
\delta t
=
O(\varepsilon^2),
\qquad
\Delta
=
O(\varepsilon),
\qquad
\varepsilon\rightarrow0.
\]

Consequently,

\[
\delta t
\ll
\Delta ,
\]

and the microscopic stochastic trajectories become rapidly oscillating
relative to the macroscopic observation scale.
The sliding-window average
\[
C_\Delta(\tau)
=
\frac1\Delta
\int_\tau^{\tau+\Delta}
Z_t\,dt
\]
therefore separates into a deterministic macroscopic component and a
rapidly fluctuating microscopic contribution.

This is the good place to discuss the origin or macro-time $\tau$ is more detail.
The macroscopic time parameter $\tau$ is introduced through the
coarse-graining procedure itself. More precisely, $\tau$ labels the
position of the averaging window
\[
[\tau,\tau+\Delta].
\]
Before passing to the continuum limit, the macroscopic evolution is
therefore naturally described by the sequence of successive
coarse-graining windows
\[
\tau_n=n\Delta,
\qquad
n=0,1,2,\ldots,
\]
so that one macroscopic time increment corresponds to one translation
of the averaging window. Accordingly, the evolution of any
coarse-grained quantity is initially represented by the finite
difference
\[
\frac{
C(\tau+\Delta)-C(\tau)
}{\Delta}.
\]
The hydrodynamic limit corresponds to the scale-separation limit
$
\varepsilon\rightarrow0^+,
$
As $\Delta\rightarrow0$, the sequence of successive averaging windows
becomes a continuous macroscopic time variable, and the finite
difference converges to the continuous derivative
\[
\lim_{\Delta\to0}
\frac{
C(\tau+\Delta)-C(\tau)
}{\Delta}
=
\frac{dC(\tau)}{d\tau}.
\]

Applying  It\^o calculus to the tensor process gives

\[
dZ_t
=
\mathcal LZ_t\,dt
+
d\xi_t ,
\]

where \(d\xi_t\) is a martingale increment satisfying

\[
\mathbb E[d\xi_t]=0.
\]

After averaging over the interval
\([\tau,\tau+\Delta]\),

\[
\frac1\Delta
\int_\tau^{\tau+\Delta}
d\xi_t
=
O(\varepsilon^{1/2}),
\]

whereas

\[
\frac1\Delta
\int_\tau^{\tau+\Delta}
\mathcal LZ_t\,dt
=
\mathcal LC_\Delta(\tau)
+
O(\varepsilon).
\]

Hence

\[
\frac{dC_\Delta}{d\tau}
=
\mathcal LC_\Delta
+
O(\varepsilon^{1/2}).
\]

Taking the limit

\[
\varepsilon\rightarrow0
\]

eliminates all residual microscopic fluctuations and yields the closed
macroscopic evolution

\[
\frac{dC}{d\tau}
=
\mathcal LC.
\]

Thus only the deterministic drift and the quadratic   It\^o variation survive
the hydrodynamic limit.

The disappearance of higher-order microscopic fluctuations explains why
the resulting macroscopic dynamics is local in time and therefore
Markovian.

\section{Why the Double Covariance Construction Produces Markovian GKSL Dynamics}

What is happening in the present framework can be understood intuitively before it is understood mathematically. The later GKSL derivation is not an independent algebraic accident added on top of the model; rather, it is already encoded in the original physical picture of scale separation and coarse-graining.

At the microscopic level, the theory assumes that the subsystem states $X_t$ and $Y_t$ evolve under rapidly fluctuating subquantum stochastic fields. These fluctuations live on an extremely short temporal scale $\delta t$, while observable physics occurs only after averaging over a much larger coarse-graining window $\Delta,$

The essential physical assumption is that the microscopic fluctuations decorrelate extremely quickly. Their memory exists only on the tiny scale $\delta t$. Once the system is observed through the larger window $\Delta$, the fine temporal structure of those fluctuations is washed out. The observer no longer sees the detailed microscopic trajectories, but only their accumulated statistical effect. This is the conceptual origin of Markovianity in the model.

The macroscopic observer does not have access to the microscopic history of the stochastic field. The coarse-graining procedure continuously erases that history. As a result, the effective evolution at the macroscopic scale depends only on the present coarse-grained covariance operator $C(\tau)$, not on the full past trajectory of the system. The dynamics therefore become local in macroscopic time.

This intuition is already visible in the earliest coarse-graining argument of the paper. The localized state-dependent fields $\xi_A(t)$ and $\xi_B(t)$ fluctuate rapidly at the microscopic level, but over short-time averaging they collapse onto effective tensor interactions. The projector identities show that the rapidly fluctuating local feedback terms become statistically equivalent to the nonlocal tensor operator: $A_m \otimes B_m)Z_t.$ At this stage, the theory is already replacing detailed microscopic structure with an averaged effective description. The local fluctuations are not tracked individually anymore; only their averaged contribution survives. This is the first appearance of the loss of microscopic memory.

The later mathematical derivation simply formalizes this physical picture rigorously.

The crucial step is the assumption that the microscopic noise correlations are delta-like:

\[
E[dW(t)dW^*(s)]
=
\Gamma \delta(t-s)\,dt\,ds.
\]

This relation means that the stochastic environment possesses no temporal memory beyond an infinitesimal instant. Correlations vanish immediately once $t \neq s$. The noise is therefore white noise in the strict stochastic sense.

Once this assumption is combined with the hydrodynamic scaling
$
\delta t = O(\varepsilon^2),
\qquad
\Delta = O(\varepsilon),
\qquad
\varepsilon \to 0^+,
$
the microscopic memory time shrinks much faster than the observational window. Consequently,
$
\frac{\delta t}{\Delta} \to 0.
$
In the hydrodynamic limit, the observer literally loses the ability to resolve microscopic temporal correlations. All finite-memory structures disappear under averaging. What survives is only the instantaneous statistical transport law governing the covariance operator.

This is precisely why the resulting equation is Markovian.

The mathematics of It\^o calculus then determines the exact form of that Markovian evolution. The quadratic variations of the Wiener processes generate the Lindblad ``sandwich'' terms, while normalization conservation produces the compensating anticommutator drift terms. Because the resulting evolution is both time-local and completely positive, the generator necessarily takes GKSL form.

Thus the appearance of the GKSL equation is not an arbitrary consequence of stochasticity alone. It arises specifically because the model combines:

\begin{enumerate}
\item rapidly decorrelating microscopic noise,
\item hydrodynamic coarse-graining,
\item conditional expectation over fast fluctuations,
\item and a strict separation between microscopic and macroscopic time scales.
\end{enumerate}

Together these mechanisms erase microscopic memory and convert the underlying stochastic dynamics into a macroscopic semigroup evolution.

In this sense, the theory is structurally analogous to ordinary hydrodynamics or kinetic theory. Molecular collisions in a fluid possess complicated microscopic histories, but after coarse-graining only local transport equations survive. Here the rapidly fluctuating subquantum trajectories play the role of microscopic molecular dynamics, while the GKSL equation emerges as the macroscopic transport equation for the covariance operator.

From this perspective, the emergence of quantum Markovianity is not a secondary feature of the model. It is the direct mathematical expression of the coarse-graining architecture itself.

Equally important is the converse observation: non-Markovian dynamics do not appear because every source of memory has been systematically removed by the assumptions of the framework. The model uses white-noise Wiener processes, instantaneous covariance structure, and a hydrodynamic limit that suppresses finite-time correlations. If one were instead to retain finite correlation times, colored noise kernels, or the full two-time covariance structure $E[Z_t Z_s^*]$ without collapsing it into a local operator, then memory effects would survive the coarse-graining procedure, and the resulting macroscopic dynamics would likely become non-Markovian rather than GKSL.

\medskip

{\bf Conflict of interest statement:} I don't have conflict of interest.

{\bf Data availability statement:} This manuscript has no associated data.

\end{document}